\documentclass[letterpaper]{article} 
\usepackage{aaai25}  
\usepackage{times}  
\usepackage{helvet}  
\usepackage{courier}  
\usepackage[hyphens]{url}  
\usepackage{graphicx} 
\usepackage{natbib}  
\usepackage{caption} 
\usepackage{newfloat}
\usepackage{listings}
\DeclareCaptionStyle{ruled}{labelfont=normalfont,labelsep=colon,strut=off} 
\usepackage{epsfig,epsf,psfrag,array,amsmath,amssymb,amsthm}
\usepackage{url,comment,enumerate,graphicx,graphics}
\usepackage{multicol}
\usepackage{pdfpages,framed,tcolorbox}

\usepackage[algoruled, vlined, linesnumbered]{algorithm2e}

\SetKw{Break}{break}
\SetKw{Continue}{continue}

\newcommand*{\nwspace}{\hspace*{.1em}}

\renewcommand{\leq}{\leqslant}
\renewcommand{\geq}{\geqslant}
\renewcommand{\le}{\leqslant}
\renewcommand{\ge}{\geqslant}
\renewcommand{\epsilon}{\varepsilon}

\let\oldsqrt\sqrt
\def\hksqrt{\mathpalette\DHLhksqrt}
\def\DHLhksqrt#1#2{\setbox0=\hbox{$#1\oldsqrt{#2\,}$}\dimen0=\ht0
   \advance\dimen0-0.2\ht0
   \setbox2=\hbox{\vrule height\ht0 depth -\dimen0}%
   {\box0\lower0.4pt\box2}}
\renewcommand\sqrt\hksqrt

\newtheorem{theorem}{Theorem}
\newtheorem{lemma}[theorem]{Lemma}

\newtheorem{corollary}[theorem]{Corollary}

\newcommand{\poly}{\textsf{poly}}

\newcommand{\prob}{\operatorname{P}}

\newcommand*{\Otilde}{\widetilde{O}}

\newcommand{\F}{{\cal F}}

\providecommand{\ignore}[1]{} 

\title{Efficient Fault-Tolerant Search by Fast Indexing of Subnetworks}
\author {
    Davide Bilò\textsuperscript{\rm 1},
    Keerti Choudhary\textsuperscript{\rm 2},
    Sarel Cohen\textsuperscript{\rm 3},
    Tobias Friedrich\textsuperscript{\rm 4},
    Martin Schirneck\textsuperscript{\rm 5}
}
\affiliations {
    \textsuperscript{\rm 1}Department of Information Engineering, Computer Science and Mathematics, 
	University of L'Aquila, Italy\\
    \textsuperscript{\rm 2}Department of Computer Science and Engineering, Indian Institute of Technology Delhi, India\\
    \textsuperscript{\rm 3}School of Computer Science, The Academic College of Tel Aviv-Yaffo, Israel\\
    \textsuperscript{\rm 4}Hasso Plattner Institute, University of Potsdam, Germany\\
    \textsuperscript{\rm 5}Faculty of Computer Science, University of Vienna, Austria\\
    davide.bilo@univaq.it, keerti@iitd.ac.in, sarelco@mta.ac.il, tobias.friedrich@hpi.de, martin.schirneck@univie.ac.at
}

\begin{document}

\maketitle

\begin{abstract}
We design sensitivity oracles for error-prone networks.
For a network problem $\Pi$, the data structure
preprocesses a network $G=(V,E)$ and sensitivity parameter $f$ such that,
for any set $F\subseteq V\cup E$ of up to $f$ link or node failures,
it can report a solution for $\Pi$ in $G{-}F$.
We study three network problems $\Pi$.
\begin{itemize}
    \item $L$-\textsc{Hop Shortest Path}: Given $s,t \in V$, is there a shortest $s$-$t$-path in $G{-}F$ with at most $L$ links? 
    \item $k$\textsc{-Path}: Does $G{-}F$ contain a simple path with $k$ links?
    \item \textsc{$k$-Clique}: Does $G{-}F$ contain a clique of $k$ nodes?
\end{itemize}

Our main technical contribution is a new construction of $(L,f)$-\emph{replacement path coverings}
($(L,f)$-RPC) in the parameter realm where $f \,{=}\, o(\log L)$.
An  $(L,f)$-RPC is a family $\mathcal{G}$ of 
subnetworks of $G$ which, for every $F \,{\subseteq}\, E$ with $|F| \,{\le}\, f$, 
has a subfamily $\mathcal{G}_F \,{\subseteq}\, \mathcal{G}$ such that (i) no subnetwork in $\mathcal{G}_F$ contains a link of $F$ and (ii) for each $s,t \,{\in}\, V$, if $G{-}F$ contains a shortest $s$-$t$-path with at most $L$ links, then some subnetwork in $\mathcal{G}_F$ retains at least one such path. Our $(L, f)$-RPC has almost the same size as the one by \citet{WY13} but it improves the time to query $\mathcal{G}_F$ from $\Otilde(f^2L^f)$ to $\Otilde(f^{\frac{5}{2}} L^{o(1)})$. 
It also improves over the size and query time of the $(L,f)$-RPC by \citet{KarthikParter21DeterministicRPC} by nearly a factor of $L$.
We then derive oracles for $L$-\textsc{Hop Shortest Path}, \textsc{$k$-Path}, and \textsc{$k$-Clique} from this. Notably, our solution for \textsc{$k$-Path} improves the query time of the one by \citet{Bilo22FixedParameterSensitivityOracles} for $f=o(\log k)$.
\end{abstract}

\begin{table*}[t]
\centering
\begin{tabular}{ccccc}

\textbf{Total Subnetworks} & \textbf{Query Time} & \textbf{Subnetworks Relevant to }$F$ & \textbf{Randomization} & \textbf{Reference} \\
\noalign{\hrule height 1pt}\\[-8pt]

 $\Otilde(fL^f)$ & $\Otilde(f^2 L^f)$ & $\Otilde(f L^{f-|F|})$ & randomized & \citet{WY13}\\[.25em]

 $\Otilde(fL)^{f+1}$ & $\Otilde(f^2 L)$ & $\Otilde(fL)$ & deterministic & \citet{KarthikParter21DeterministicRPC} \\[.25em]
 
 $\Otilde(fL^{f+o(1)})$ & $\Otilde(f^{\frac{5}{2}} L^{o(1)})$ & $\Otilde(fL^{o(1)})$ & randomized & Theorem~\ref{thm:stand-alone-trees}
\end{tabular}
\caption{
	Comparison of $(L,f)$-replacement path coverings.
	The sensitivity $f$ and cut-off parameter $L$ satisfy $f = o(\log L)$.
	The number of subnetworks relevant for the failure set $F$ refers to the size $|\mathcal{G}_F|$,
	while the query time is the time to compute $\mathcal{G}_F$.
 }
\label{table:results_trees}
\end{table*}

\section{Introduction}

Networks are central structures in computer science as they can model different types of relations 
we encounter in real-world applications. 
Numerous algorithms and data structures have been developed to solve problems in static networks where nodes and links do not change over time. However, as networks in real life are prone to transient failures, many of these algorithms require recomputations from scratch even when only a few network components  are malfunctioning. In many applications one may have an a priori known bound $f$ on the number of simultaneous failures, especially in the context of independent failures where the likelihood of multiple failures decreases exponentially.
This is called the \emph{fault-tolerant} setting or \emph{sensitivity} analysis
and the resulting data structures are called \emph{sensitivity oracles}.
In the past two decades, 
many sensitivity oracles have been designed for
classical network problems, e.g.,\ connectivity~\cite{PatrascuT:07,DuanP09a, DuanP10, DuanP:17}, 
shortest paths~\cite{
ChoS024, 
DeThChRa08,CLPR12,DuanP09a, BiloCCC0KS23, DuRe22, DeyGupta24, GuR21} and routing~\cite{CLPR12}. 

We continue this line of work by designing data structures that, given a network problem $\Pi$, 
a network $G=(V,E)$ with $n$ nodes and $m$ links, and a parameter $f$,
preprocesses $G$ into a 
\emph{oracle with sensitivity $f$} that, when queried with a set $F\subseteq V\cup E$ of size  $|F|\leq f$, reports a solution for $\Pi$ in $G{-}F$. Specifically, we study the following problems $\Pi$.
\begin{itemize}
    \item $L$-\textsc{Hop Shortest Path}: Given two nodes $s,t \in V$, is there a shortest $s$-$t$ path in $G{-}F$ with at most $L$ links?
    \item $k$-\textsc{Path}: Does $G{-}F$ contain a simple path with $k$ links? 
    \item \textsc{$k$-Clique}: Does $G{-}F$ contain a clique of $k$ nodes? 
\end{itemize}
In fact, we can handle all problems with the property that any certificate solution $S\subseteq G$ for $\Pi$, like e.g., a $k$-path, is also a certificate solution of $G'$ for $\Pi$, for every $S\subseteq G' \subseteq G$. Although we focus on decision problems, our data structures are also capable of reporting a certificate solution.

A naive solution consists in enumerating  all possible subnetworks, one for each $F\subseteq V\cup E$ with  $|F|\leq f$, and storing a static data structure for each computed subnetwork. This is, however, prohibitively expensive w.r.t. both space and preprocessing time as there are 
$\binom{|V|+|E|}{\le f} = \Theta(n^f \,{+}\, m^f)$ 
subnetworks if node and link failures can occur.
We need a more efficient way to construct subnetworks that allow us to make the data structure fault tolerant.
Moreover, we also need some compact and time-efficient indexing scheme to quickly access the correct subnetworks upon query.

Our motivation aligns with the work of \cite{ChoS024}, which addresses fault tolerance in dynamic settings such as navigation, logistics, and communication networks, where rapid rerouting is crucial. Similarly, \cite{Blum_Funke_Storandt_2018} motivates using preprocessing to reduce the search space for path finding. Many real-world networks have a small diameter, and this makes our work relevant as we assume only the existence of all-pairs shortest paths of length bounded by a parameter $L$. 
The works by \citet*{OuyangYQCZL20} and \citet{Zhang0021} align with our aim to enhance robustness in AI systems. We provide combinatorial insights that might be useful to design new time-efficient heuristics.

\paragraph{$(L,f)$-Replacement Path Coverings.}
\citet{WY13} design the first data structure with sensitivity $f$ 
for \textsc{$L$-Hop Shortest Path}.
They used it as a building block in a \emph{distance sensitivity oracle} ($f$-DSO),
which reports length of general shortest paths in $G{-}\F$ without the hop-length constraint.
Their main tool was an $(L,f)$-\emph{replacement path covering},\footnote{
 	The name was introduced later by
 	\citet{KarthikParter21DeterministicRPC}.
 } or $(L,f)$-RPC, a family $\mathcal{G}$ of 
subnetworks of $G$ which, for every set $F \subseteq E$ with $|F| \le f$, has a subfamily $\mathcal{G}_F \subseteq \mathcal{G}$ such that (i) no subnetwork in $\mathcal{G}_F$ contains an edge in $F$ and (ii) for every pair of nodes $s,t \,{\in}\, V$, if $G{-}F$ contains an $L$-hop shortest $s$-$t$-path, then some subnetwork in $\mathcal{G}_F$ retains such a path.
That means, to find the true $s$-$t$-distance in $G{-}F$,
it is enough to consult the subnetworks in $\mathcal{G}_F$.
\citet{WY13} obtained their $(L,f)$-RPC using $O(fL^f \log n)$ copies of $G$ and, in each one,
removing any link independently with probability $1/L$.
Chernoff's bound shows that w.h.p.\footnote{
	We say an event occurs \emph{with high probability} (w.h.p.)
	if it has probability at least $1-n^c$ for a constant $c > 0$.
} there are $|\mathcal{G}_F| = O(fL^{f-|F|} \log n)$
subnetworks that do not contain any link of $F$ and such that, for every $s,t \in V$, at least one element in $\mathcal{G}_F$ retains an $L$-hop shortest $s$-$t$ path in $G-F$, if there is any.
Unfortunately, due to the randomness, the most efficient way to find $\mathcal{G}_F$ 
 is to go through all subnetworks in $\mathcal{G}$ individually,
taking time $O(f^2 L^f \log n)$.

\citet{KarthikParter21DeterministicRPC} derandomized this construction
via error-correcting codes. For $L \,{\ge}\, f$, their $(L, f)$-RPC has size $O(fL \log n)^{f+1}$, which is an \mbox{$O(f^f L \ (\log n)^f)$-factor} larger than the solution by Weimann and Yuster. 
The size of the relevant subfamily $\mathcal{G}_F$ is $\Otilde(fL)$,\footnote{The
$\Otilde(\cdot)$ notation hides
poly-logarithmic factors in the number $n$ of nodes of the network.
} independently of $|F|$.
The huge advantage of the deterministic construction is the much better query time.
They showed how to retrieve $\mathcal{G}_F$ in time $\Otilde(f^2L)$. 
A side-by-side comparison is given in Table~\ref{table:results_trees}.

\begin{table*}[t]
\centering
\begin{tabular}{cccc}

\textbf{Query Time} & \textbf{Space} & \textbf{Preprocessing Time} & \textbf{Reference} \\
\noalign{\hrule height 1pt}\\[-8pt]

 $\Otilde\!\left(\!\left(\frac{f+k}{f}\right)^f \!\left(\frac{f+k}{k}\right)^k \! fk \!\right)$ & 
 	$\Otilde\!\left(\!\left(\frac{f+k}{f}\right)^f \!\left(\frac{f+k}{k}\right)^k \! fk \!\right)$ & 
 	$\left(\frac{f+k}{f}\right)^f \!\left(\frac{f+k}{k}\right)^k \! f \cdot 2^k \, \poly(n)$ & \citet{Bilo22FixedParameterSensitivityOracles}\\[.25em]

 $f^2 \nwspace 2^k \, \poly(k)$ & $\Otilde(2^k \cdot n^2)$ & $2^k \, \poly(k) \cdot n^\omega$ & \citet{AlmanHirsch22ExteriorAlgebras} \\[.25em]
 
 $\Otilde\!\left(4^f f^{2-o(1)} \nwspace k^{o(1)}\right)$ & $\Otilde\!\left(\left(\frac{k+4}{f}\right)^{f+o(1)} f^{3/2} \nwspace k\right)$ & $\left(\frac{k+4}{f} \right)^{f+o(1)} f^{3/2} \cdot 2^k \,\poly(n)$ & Theorem~\ref{thm:k-path}
\end{tabular}
\caption{Comparison of fixed-parameter sensitivity oracles for the $k$-\textsc{Path} problem. 
	All results are randomized. The last row assumes $f = o(\log k)$. 
	The quantity $\omega < 2.371552$ is the matrix multiplication exponent.}
\label{table:results_k-path}
\end{table*}

While the randomized $(L, f)$-RPC is smaller than the deterministic one,
the latter has a much better query time. 
Thus, a natural question is whether one can design an $(L, f)$-RPC whose size is as small as the former and whose query time is at least as efficient as the latter.
We answer this question affirmatively.
The bottleneck of the construction by \citet{WY13} is not randomness,
but \emph{independence}.
This is the reason one has to scan all subnetworks in order to find $\mathcal{G}$.
Our construction can be seen as an indexing of the Weimann and Yuster graphs.
The crucial difference is that the subnetworks we generate are no longer independent.
Instead, the construction process naturally groups them into \emph{sampling trees}.
The query algorithm merely traces a root-to-leaf path in the tree by always choosing an arbitrary
child node whose subnetwork contains no link of $F$. 

\begin{theorem}
\label{thm:stand-alone-trees}
Let $G$ be a directed/undirected network with $n$ nodes, possibly weighted, and let 
$f = o(\log L)$. We can build a randomized $(L,f)$-replacement path covering $\mathcal{G}$ of size 
$\Otilde(fL^{f+o(1)})$ such that, given any $F\subseteq V \cup E$ with $|F| \leq f$, computes in time $\Otilde(f^{\frac{5}{2}}L^{o(1)})$ a collection $\mathcal{G}_F \subseteq \mathcal{G}$ of subnetworks satisfying the following properties w.h.p.: 
\begin{enumerate}[($i$)]
    \item $|\mathcal{G}_F|=\Otilde(fL^{o(1)})$;
    \item No subnetwork in $\mathcal{G}_F$ contains an element of $F$;
    \item For any two nodes $s,t \in V$ that admits an $L$-hop shortest $s$-$t$ path in $G-F$, there is at least one subnetwork in $\mathcal{G}_F$ that retains one of such paths.
	\end{enumerate}
\end{theorem}

The number of subnetworks in our construction is only an $L^{o(1)}$-factor away
from the one by \citet{WY13}. The $\Otilde(f^{5/2}L^{o(1)})$-query time 
is even better then the one by \citet{KarthikParter21DeterministicRPC}.
The assumption $f = o(\log L)$ is to ensure that parameters of the $(L,f)$-RPC
scale only logarithmically in the size of the input network $G$
(hidden in the $\Otilde$-notation).
Our construction also works for larger sensitivities like, e.g., $f = o(\log(n)/\log\log n)$,
a common bound in the literature \cite{WY13}, \cite*{ChCoFiKa17},
\cite*{Bilo24AImprovedDO},
and even up to $f = o(\log n)$.
However, then the total number of subnetworks, the number of relevant subnetworks $|\mathcal{G}_F|$,
and the query time all increase by an $n^{o(1)}$-factor.

Next, we apply Theorem~\ref{thm:stand-alone-trees} to \textsc{$L$-Hop Shortest Path}.
\begin{theorem}
\label{thm:short-distance-exact}
	Let $G$ be a (directed) network with $n$ nodes and real edge weights that does not contain negative cycles.
	Let 
 $f = o(\log L)$.
	There exists a randomized $L$-hop distance  oracle with sensitivity $f$
	for pairwise $L$-hop shortest paths 
	that takes space $\Otilde(f L^{f+o(1)} n^2)$ 
	and has query time $\Otilde(f^{5/2}L^{o(1)})$.
	The oracle can be preprocessed in time
	$\Otilde(f L^{f+o(1)} \, T_{\text{APSP}})$,
	where $T_{\text{APSP}}$ is the time to compute all-pairs shortest paths in $G$. 
\end{theorem}

It has often been observed that real-world networks are modeled well by networks with small diameter, see the works of~\citet{watts_collective_1998, Albert1999Diameter, 
Adamic99thesmall}, the ``small-world'' networks by Kleinberg~\shortcite{kleinberg2000navigation}, Chung-Lu networks~\cite{chung2002average},
hyperbolic random networks~\cite{friedrich2018diameter}, and the preferential attachment model~\cite{Hofstad16RandomGraphs}.
Let $D$ be the diameter of an undirected, unweighted network $G$.
By the result of \citet*{Afek02RestorationbyPathConcatenation_journal} showing that
if $G{-}F$ is still connected, then its diameter is at most $(f{+}1)D$,
Theorem~\ref{thm:short-distance-exact} gives a very efficient $f$-DSO
for \emph{general} hop-lengths in networks with, say, polylogarithmic diameter. 

\begin{corollary}
\label{cor:bounded-diam-exact}
	Let $G$ be an undirected, unweighted network with $n$ nodes, $m$ links, and diameter $D = \omega(1)$.
	For any constant $f$, there exists a randomized distance oracle with sensitivity $f$
	that takes space $\Otilde(D^{f+o(1)} n^2)$ 
	and has query time $\Otilde(D^{o(1)})$.
	The oracle can be preprocessed in time
	$\Otilde(D^{f+o(1)} mn)$ or $\Otilde(D^{f+o(1)} n^\omega)$,
	where  $\omega < 2.371552$ is the matrix multiplication exponent.
\end{corollary}

\paragraph{On $k$-\textsc{Path} and $k$-\textsc{Clique}.}
The $k$-\textsc{Path} problem is \mbox{\textsf{NP}-complete} when $k$ is given as part of the input. 
If $k$ is treated as a parameter, however, then the problem turns out to be
\emph{fixed-parameter tractable} (\textsf{FPT}),
meaning that it is solvable in time $g(k) \cdot \poly(n)$ for some function $g$.
The current-best algorithms are randomized and run in time
$1.66^k \cdot \poly(n)$ for undirected networks~\cite*{Bjoerklund17NarrowSieves}
and $2^k \cdot \poly(n)$ for directed networks~\cite{Williams_2009}.
The best deterministic algorithm runs in time $2.554^k \cdot \poly(n)$~\cite{Tsur_2019}.

\citet*{Bilo22FixedParameterSensitivityOracles} introduced
\emph{fixed-parameter sensitivity oracles} for \textsf{FPT}-problems. In these oracles
the preprocessing time and space requirement must be of the form $g(f,k) \cdot \poly(n)$,
and the query time ought to be ``significantly faster''
than recomputing a solution from scratch.
They designed a fixed-parameter oracle with sensitivity $f$ for $k$-\textsc{Path} 
with a query time and space of $O((\frac{f+k}{f})^f (\frac{f+k}{k})^k fk \log n)$,
and a $(\frac{f+k}{f})^f (\frac{f+k}{k})^k f 2^k \cdot \poly(n)$ preprocessing time.
When $f \,{\ge}\, k$, 
\citet{AlmanHirsch22ExteriorAlgebras} significantly improved the query time to
$f^2 2^k \poly(k)$ randomized or $O(f^2 2^{\omega k})$ deterministic.
The space requirement is $O(2^k (\log k) n^2)$
in the randomized case and $O(2^k k n^2 \log n)$ for the deterministic version.
The preprocessing time is $2^k \poly(k) n^\omega$ (randomized)
or $4^k \poly(k) n^{\omega}$ (deterministic).
See Table~\ref{table:results_k-path} for an overview.

We consider the case where $k$ is much larger than $f$,
so our goal is to make the dependence on $k$ as small as possible.

\begin{theorem}
\label{thm:k-path}
	Let $G$ be a directed, unweighted network with $n$ nodes.
	Let $f$ and $k$ be two integer parameters with $f = o(\log k)$.
	There exists a randomized fixed-parameter oracle with sensitivity $f$ for \textsc{$k$-Path}
	that takes space $\Otilde((\frac{k+4}{f})^{f+o(1)} f^{3/2} \nwspace k)$ 
	and has query time $\Otilde(4^f f^{2-o(1)} \nwspace k^{o(1)})$ w.h.p.
	The oracle can be preprocessed in randomized time
	$(\frac{k+4}{f})^{f+o(1)} f^{3/2} \cdot 2^k \, \textup{\poly}(n)$.
	If the network is undirected, the preprocessing time decreases to
	$(\frac{k+4}{f})^{f+o(1)} f^{3/2} \cdot 1.66^k \, \textup{\poly}(n)$.
\end{theorem}

During the analysis, it becomes apparent that the path structure
is not actually needed, only the fact that the solution has $k$ links.
This allows us to extend sensitivity oracles even to networks motifs
that are believed not to have an \textsf{FPT}-algorithm,
like \textsc{$k$-Clique}.
It is widely believed \textsc{$k$-Clique} is not solvable in time $g(k) \cdot \poly(n)$
as the clique detection problem is \textsf{W}[1]-complete.
The current best algorithm by \citet{NesetrilPoljak85SubgraphProblem}
computes a clique of $k$ nodes in time $O(n^{\omega k/2})$.
We use it to design the first fixed-parameter sensitivity oracle for \textsc{$k$-Clique}.

\begin{theorem}
\label{thm:k-clique}
	Let $G$ be an undirected, unweighted network with $n$ nodes.
	Let $f$ and $k$ be two integer parameters with $f = o(\log k)$.
	There exists a randomized oracle with sensitivity $f$ for \textsc{$k$-Clique}
	that takes space $\Otilde((\frac{k^2+4}{f})^{f+o(1)} f^{3/2} \nwspace k^2)$ 
	and has query time $\Otilde(4^f f^{2-o(1)} \nwspace k^{o(1)})$ w.h.p.
	The oracle can be preprocessed
	$\Otilde((\frac{k^2+4}{f})^{f+o(1)} f^{3/2} \cdot n^{\omega k/2})$.
\end{theorem}

\paragraph{Outline.}
We first discuss replacement path coverings in general in the next section, 
followed by their most common application in $L$-hop distance sensitivity oracles. 
Afterwards, we show that similar ideas can also be employed for $k$-paths, $k$-cliques, and other graph motifs.
We conclude this work by discussing some open questions.

\section{The Main Tool: Sampling Trees}
\label{sec:sampling_trees}

We turn the idea of removing links to create an $(L,f)$-replacement path covering upside down.
We build a \emph{sampling tree} having the empty subnetwork 
as its root and, in each level, any child node takes all the links of its parent
and randomly \emph{adds} new links from $G$ with some probability depending on $f$ and $L$, and on an additional parameter $h$
that controls the height of the trees.
The probability is fine-tuned in such a way that, among the leaves, the likelihood for any link to exist
is precisely the same as it was in the construction of \citet{WY13}.
This brings the number of subnetworks back down to $O(fL^{f+o(1)} \log n)$.

The crucial difference is that the stored subnetworks are no longer independent.
Two nodes of the same tree always share at least all links that are stored in their lowest common ancestor.
This arrangement of the subnetworks allows for a very efficient query algorithm
even though the construction is randomized.
Given a query $F\subseteq E$ with $|F|\leq f$, we navigate the sampling tree starting from the root and, in each level, we move to an \emph{arbitrary} child node whose stored subnetwork does not have any link of $F$.
The flip side of such a simple procedure is that the leaf node we reach may be relevant for $F$ only with a small probability; as it will turn out, exponentially small.
In turn, it depends only on the parameter $h$ and not on $f$, $L$, or the network size.
Optimizing $h$ and repeating the query in sufficiently many independent trees
ensures a high success probability.

\paragraph{Detailed Construction.}
Let $f, L$ be positive integers that may depend on $n$.
The \emph{sensitivity} $f$ is assumed to be much smaller than the cut-off parameter $L$,
namely, $f  = o(\log L)$.
We now implement the data structure that preserves paths with hop-length most $L$
against up to $f$ failing links or nodes in the network.
We focus first on fault tolerance against \emph{link failures}
and later extend it to node failures.
For now,
let $F \subseteq E$ with $|F|\leq f$.
We construct a collection of $K$ sampling trees whose nodes all hold a subnetwork of $G$.
Each sampling tree has height $h$ and any internal node has exactly $\alpha$ children.
$K$, $h$, and $\alpha$ are parameters to be optimized later.

A single tree has $\alpha^h$ leaves and $O(\alpha^h)$ nodes in total.
We associate with each node $x$ a set $A_x \subseteq E$.
Intuitively, $A_x$ are the links that are \emph{missing} in the network stored in $x$;
equivalently, node $x$ holds the network $G_x = G{-}A_x$.
If $x$ is the root of the tree, we set $A_x = E$ so that the corresponding network 
is the empty network (on the same node set $V$).
Now let $y$ be a child of some node $x$,
its set $A_y \subseteq A_x$ is obtained by selecting each link in $A_x$
independently with probability $p$.
This construction is iterated until the tree has height $h$.
In the same fashion, we build all the sampling trees $T_0, T_1, \dots, T_{K-1}$ 
where each random choice along the way is made independently of all others. The total number of stored subnetworks is $O(K \alpha^h)$ and the $f$-covering of $G$ is given by the family $\mathcal{G}$ of all subnetworks that are stored in the leaves of all the $K$ sampling trees. 

We will heavily use the fact that
the link distribution in the random sets $A_x$ only depends on the depth $r$ of the node $x$.
Moreover, even though two different nodes are not independent,
the links present in a single node indeed are.
The following lemma has a simple proof by induction over $r$.

For a positive integer $\ell$, we denote $\{0,\ldots,\ell-1\}$ by $[\ell]$.

\begin{lemma}
\label{lem:tree-property}
	Let $i \in [K]$ be an index and $x$ a node of the tree $T_i$ at depth $r \ge 0$.
	For any $e \in E$, we have $\prob[e\in A_x]=p^r$.
	Moreover, for any two different links $e,e' \in E$, 
	the events $e \in A_x$ and $e' \in A_x$ are independent.
\end{lemma}

\begin{algorithm}[t!b]
\caption{Query algorithm.}
\label{alg:tree-exact}
	$\mathcal{G}_F \gets \emptyset$\;
	\For{$i=0$ \KwTo $K{-}1$}
	{
		$x \gets$ root of $T_i$\;
		\While{$x$ is not a leaf}
		{
			\emph{noChildFound} $\gets$ \textsc{true}\;
			\ForAll{children $y$ of $x$}
			{
				\If{$F \subseteq A_y$ \label{line:containment}}
				{
					$x \gets y$\;
					\emph{noChildFound} $\gets$ \textsc{false}\;
					
					\Break inner for-loop\;
				}
			}
			\If{noChildFound}
			{
				\Continue outer for-loop\; \label{line:continue}
			}
		}
		$\mathcal{G}_F \gets \mathcal{G}_F \cup \{G_x\}$\; \label{line:result}
	}
\end{algorithm}

\paragraph{Query Algorithm.}
For the data structure to be efficient,  we need to quickly find $\mathcal{G}_F$ 
when given a query set $F\subseteq E$ with $|F|\leq f$.
The procedure is summarized in Algorithm~\ref{alg:tree-exact}.
Recall that the set $A_x \subseteq E$ contains those links that are \emph{removed} 
in the subnetwork $G_x$ of the node $x$. 
The trees $T_0, \dots, T_{K-1}$ are searched individually, starting in the respective roots.
In each step, the algorithm always chooses some (any) child $y$ of the current node $x$ 
with $F \subseteq A_x$, and continues the search there.
If no such $y$ exists, the current tree is abandoned.
Once a leaf is reached, the subnetwork stored there is added to $\mathcal{G}_F$.
Up to $K$ subnetworks relevant for the query set $F$ are collected in total time $O(fK\alpha h)$
as every node has $\alpha$ children and the trees have height $h$.

\paragraph{Analysis.}
Before we optimize the parameters $K$, $\alpha$, $h$, and $p$,
we show the correctness of the query algorithm.
It is clear that no network in $\mathcal{G}_F$ contains a failing link from $F$
since it is explicitly verified that $F \subseteq A_y$ before recursing to $y$.
It could be that $\mathcal{G}_F$ remains empty.
This happens if, for every single tree,
the outer for-loop of Algorithm~\ref{alg:tree-exact} is continued in line~\ref{line:continue}
since a node is encountered whose children all have $F \nsubseteq A_y$.
As part of the correctness proof, we bound the probability of this event. In the following, given two nodes $s,t \in V$ and $F\subseteq E$, we denote by $\pi(s,t,F)$ a shortest path from $s$ to $t$ in $G-F$, a.k.a. a \emph{replacement path}.

\begin{lemma}
\label{lem:parent-to-child}
	Let $i \in [K]$, $F \subseteq E$ with $|F| \le f$,
	and  $s,t \in V$  such that $\pi=\pi(s,t,F)$ contains at most $L$ links.
	\begin{enumerate}[($i$)]
		\item Algorithm~\ref{alg:tree-exact} reaches a leaf of the sampling tree $T_i$
			with probability at least $((1-(1-p^f)^{\alpha})^h$.
		\item If Algorithm~\ref{alg:tree-exact} reaches a leaf $x$ of $T_i$,
			then the probability of $\pi$ 
            existing in $G_x$ is at least $(1-p^h)^L$.
	\end{enumerate}
\end{lemma}

\begin{proof}
	To prove Clause~$(i)$, we first establish the following claim.
	Let $x$ be an inner node of $T_i$ with $F \,{\subseteq}\, A_x$
	and $y_0, \dots, y_{\alpha-1}$ its children.
	Then, the probability that a child has $F \,{\subseteq}\, A_{y_j}$ 
	is at least $1{-}(1{-}p^f)^\alpha$.
	Any set $A_{y_j}$ is obtained from $A_x$ by sampling each link independently with probability $p$.
	Since $F\,{\subseteq}\, A_x$, we have $\prob[F \,{\subseteq}\, A_{y_j}] = p^{|F|} \ge p^f$ and
	$\prob[\, \exists j \in [\alpha] \colon F \subseteq A_{y_j}]
			=1-\prod_{j=0}^{\alpha-1} \prob[F \nsubseteq A_{y_j}]$.
	So the latter probability is lower bounded by $1-\left(1-p^f \right)^{\alpha}$.
	
	The derivation depends only the condition $F \subseteq A_x$ and not on the path through $T_i$
	by which the query algorithm reaches the node $x$.
	Since Algorithm~\ref{alg:tree-exact} maintains this condition,
	we can iterate that argument for each of the $h$ parent-child transitions,
	which proves Clause~$(i)$.
	
	For Clause~$(ii)$, consider the algorithm reaching a leaf $x$ of $T_i$ at depth $h$.
	Evidently, we have $F \subseteq A_x$, but for all other links $e \in E{\setminus}F$,
	it holds that $\prob[e \in A_x] = p^h$ by Lemma~\ref{lem:tree-property}.
	The replacement path $\pi$ survives in $G_x$ with probability
	$\prob[\nwspace E(\pi) \cap A_x = \emptyset] = (1-p^h)^{|E(\pi)|} \ge (1-p^h)^L$.
\end{proof}

\paragraph{Optimizing the Parameters.}
The number of subnetworks is $O(K \alpha^h)$ out of which the query algorithm selects at most $K$
in time $O(fKah)$. This incentives us to choose all those parameters as small as possible, 
especially the height $h$ of the sampling trees.
Moreover, the probability to reach a leaf of a tree is exponentially small in $h$ 
(Lemma~\ref{lem:parent-to-child}~($i$)).
However, once a leaf is actually reached, the probability of the stored network holding
a relevant replacement path $\pi(s,t,F)$ \emph{grows} exponentially with $h$.
We need to cover \emph{all} pairs $s,t \in V$ that have a shortest path in $G{-}F$ with at most $L$ links.
Our strategy for setting the parameters is to keep the height small and instead 
boost the success probability by choosing a larger branching factor $\alpha$ 
and number of independent trees $K$, and balance the selection with a suitable sampling probability $p$.
Let $c \,{>}\, 0$ be a sufficiently large constant.
We set the parameters as
$h = \sqrt{f \ln L}$, $K = c \nwspace (\tfrac{e}{e-1})^h \nwspace f \ln n$, $\alpha = L^{f/h}$, and $p = L^{-1/h}$.

Lemma~\ref{lem:tree-property} states that in any leaf $x$, at depth $h$,
the probability for any link to be removed (i.e., $e \in A_x$) is $p^h = 1/L$.
Not coincidentally, this is the same probability used by \citet{WY13}.
We verify next that the query algorithm indeed finds a suitable collection of networks.
The lemma also implies that the networks stored in the leaves of the trees
form an $(L,f)$-replacement path covering w.h.p.

\begin{lemma}
\label{lem:correctness}
	W.h.p.\ over all $F \subseteq E$ with $|F| \le f$
	and $s,t \in V$ with $|E(\pi(s,t,F))| \le L$,
	after the termination of Algorithm~\ref{alg:tree-exact},
	the path $\pi(s,t,F)$ exists in a network of $\mathcal{G}_F$. 
\end{lemma}

\begin{proof}
	The proof heavily relies on the estimates derived in Lemma~\ref{lem:parent-to-child}.
	The probability to reach a leaf $x$ in some tree $T_i$ whose corresponding network $G_x$
	contains the replacement path $\pi(s,t,F)$ is at least $((1-(1-p^f)^{\alpha})^h \cdot (1-p^h)^L$.
	We estimate the two factors separately starting with the second.
	Inserting the parameters gives $(1-p^h)^L \ge (1-L^{-1})^L \ge \frac{1}{4}$.
	Observe that 
	$1- \left(1-p^f \right)^{\alpha} = 1-\left(1-\frac{1}{L^{f/h}} \right)^{L^{f/h}} 
			\ge 1- \frac{1}{e}$.
	The total probability is thus at least
	$\frac{1}{4}(1{-}\frac{1}{e})^h = \frac{1}{4}(\frac{e-1}{e})^h$.
	Repeating the query in $K = c \nwspace (\tfrac{e}{e-1})^h \nwspace f \ln n$
	independent trees reduces the \emph{failure} probability for any triple $(s,t,F)$ to
	$\left(1- \frac{1}{4}\left(\frac{e-1}{e} \right)^h \right)^{\! c(\tfrac{e}{e-1})^h f \ln n}
			\le e^{-\frac{c}{4} f \ln n} = n^{-\frac{c}{4}f}$.
	A union bound over the at most $|V^2 \times \binom{E}{\le f}| = O(n^{2+2f})$ triples 
	shows that the failure probability of the whole algorithm is of order 
	$O(n^{2+2f-\frac{c}{4}f})$.
	Choosing a sufficiently large constant $c$ thus ensures a high success probability.
\end{proof}

We are left to compute the number of networks and query time. 
Let $C$ abbreviate $\frac{e}{e-1} \approx 1.582$
and note that
	$C^h = C^{\sqrt{f \ln L}} = \big(L^{\frac{\ln C}{\ln L}} \big)^{\sqrt{f \ln L}}
		= L^{\ln C \sqrt{\frac{f}{\ln L}}} = L^{o(1)}$.
The last estimate uses $f = o(\log L)$.
Similarly, we have $\alpha = L^{f/h} = L^{\sqrt{f/\ln L}} = L^{o(1)}$.
Our choice of parameters thus implies that the whole data structure stores 
$O(K \alpha^h) = O(f C^h L^f \log n) = O(f L^{f+o(1)} \log n)$ networks
and computes a subfamily of $K = O(fL^{o(1)} \log n)$ of them
that are relevant for the failure set $F$ in time 
$O(fK\alpha h) = O(f (L^{o(1)})^2 f (\log n) \sqrt{f \log L}) = O(f^{5/2} L^{o(1)} \log n)$.
This proves Theorem~\ref{thm:stand-alone-trees} for the case of link failures.

If the sensitivity is up to $f = o(\log n)$,
then we have $h = \sqrt{f \ln L} = o(\log n)$ since $L$ is at most $n$.
This results in a total number of subnetworks of $O(f C^h L^f \log n) = L^f n^{o(1)}$,
$|\mathcal{G}_F| = O(f C^h \log n) = n^{o(1)}$ of which are relevant for a given query,
and a query time of $O(f^2 C^h L^{f/h} h) = n^{o(1)}$.

\paragraph{Node Failures.}
The changes needed to accompany node failures are minuscule.
Instead of a set of links, we now associate with every tree node $x$ a set $A_x \subseteq V$ of network nodes.
We have $A_x = V$ in the roots and in each child $y$ of $x$, we include any element of $A_x$
independently with probability $p$.
Note that then the network $G_x = G - A_x$ is the subnetwork of $G$ \emph{induced}
by the node set $V{\setminus}A_x$.
As the sampling remains the same,
we get an analog of Lemma~\ref{lem:tree-property}.

\begin{lemma}
\label{lem:tree-property_node_failure}
	Let $i \in [K]$ be an index and $x$ a node of the tree $T_i$ at depth $r \ge 0$.
	For any $v \in v$, we have $\prob[v\in A_x]=p^r$.
	For any two different network nodes $v,v' \in V$,
	the events $v \in A_x$ and $v' \in A_x$ are independent.
\end{lemma}

The other lemmas follow from this almost verbatim as before.
The only differences are that, if Algorithm~\ref{alg:tree-exact} reaches a leaf,
then the probability of the replacement path existing in the subnetwork
is $(1-p^h)^{L+1}$ for if the replacement path has up to $L$ links,
it can have up to $L+1$ nodes.
This changes the correction factor in Lemma~\ref{lem:correctness}
to $(1-L^{-1})^{L+1} \ge \frac{1}{8}$,
which is counter-acted by choosing the constant $c$ in the definition
of the number of trees $K$ marginally larger.
In contrast to link failures, there are $|V^2 \times \binom{V}{\le f}| = O(n^{2+f})$
relevant queries.

\section{Distance Sensitivity Oracles}

Recall that, given $s,t \in V$ and a set $F$ of at most $F$ failures, an $L$-hop $f$-DSO reports an overestimate of the length of $\pi(s,t,F)$ that matches the lower bound if $G-F$ contains an $L$-hop shortest $s$-$t$ path. It is straightforward to turn our generic tree structure of Theorem~\ref{thm:stand-alone-trees} into an $L$-hop $f$-DSO.
After all, that was the original purpose of \citet{WY13}
when defining $(L,F)$-RPCs.
The only difference we make is to precompute for every leaf $x$ of all sampling trees  the pairwise distances of nodes in $G_x$.
The query works exactly as Algorithm~\ref{alg:tree-exact}
only that, upon query $(s,t,F)$, in line~\ref{line:result} the stored \emph{distance} from $s$ to $t$ in $G_x$ is recorded.
The value returned by our $L$-hop $f$-DSO is given by the minimum of the computed distances.

The construction of the oracle is dominated by preparing the distances in the leaves.
The preprocessing time is $\Otilde(f L^{f+o(1)} \cdot T_{\text{APSP}})$,
where $T_{\text{APSP}}$ is the time needed to compute all-pairs shortest paths (APSP)
in an arbitrary subnetwork of $G$.
It depends on the properties of $G$.
If there are no negative cycles (e.g., because all link weights are non-negative),
one can use Dijkstra's or Johnson's algorithm running in time $T_{\text{APSP}} = \Otilde(mn)$.
If instead one is willing to use fast matrix multiplication,
then APSP can be computed in time $\Otilde(Mn^{\omega})$
for undirected networks with non-negative integer link weights in $[M]$,
or $O(Mn^{2.5286})$ for directed networks with integer weights in $\{-M, \dots M\}$
\cite{Alon97ExponentAPSP, Seide95APSPUnweightedUndirected, ShZw99, Zwick02DirectedAPSP}.
Here, $\omega < 2.371552$ is the matrix multiplication exponent
\cite*{Duan23FMMviaAsymtericHashing,VWilliams24MatrixMultiplicationAlphaToOmega}.

The space of the data structure is $\Otilde(f L^{f+o(1)} n^2)$,
again dominated by storing the distances in the leaves.
In turn, the query time remains at $\Otilde(f^{5/2}L^{o(1)})$
as the values $d_{G_x}(s,t)$ can be looked up in constant time.
This proves Theorem~\ref{thm:short-distance-exact}.

\section{Sensitivity Oracles for $k$-\textsc{Path}}
\label{sec:k_path}

We now turn to the fixed-parameter sensitivity oracle for the \textsf{NP}-complete $k$-\textsc{Path} problem.
We only treat link failures here, (i.e., $F \subseteq E$) to ease notation.
We further assume that access to an algorithm that computes simple paths with $k$ links ($k$-paths) in subnetworks of $G$ in a way that respects a certain
\emph{inheritance property}.
Suppose $F$ has at most $f$ links, $G{-}F$ has a $k$-path 
and the algorithm we use produces such a path $P$.
Then for any subnetwork $H \subseteq G{-}F$ that still contains $P$,
we require that the same path $P$ is also the output on of the algorithm in $H$.
Note that such a tie-breaking scheme is obtained by assign distinct weights 
to the links in $E$ and always choosing the $k$-path of minimum weight.

Our data structure is based on the sampling trees above,
where we naturally set $L = k$ and thus assume $k = o(\log k)$.
The construction of the trees with parameters $K,\alpha,h$, and $p$ is almost as before.
The difference is that any node $x$ of a tree $T_i$ is not only associated with one set of links $A_x$,
but also with a second one $S_x$.
In a parent-child traversal from $x$ to $y$,
$A_y$ is still obtained from $A_x$ by sampling each link of $A_x$ with probability $p$.

The construction of $S_y$ is a bit more involved.
Let $r$ be the depth of the parent, and thus $r+1$ the depth of the child $y$.
Let $S_x$ be the set associated with its parent.
To unify the exposition, if $y$ is the root, we set $S_x = E$.
To construct $S_y$ from $S_x$, we first build a family $\mathcal{P}_y$ of paths in $4^f \alpha^{h-(r+1)} \ln h$ independent rounds.
In each round, we sample a subset $I \subseteq A_y$ by selecting each link in $A_y$
independently with probability $p^{h-(r+1)}/2$.
If there exists a $k$-path in the network with link set $S_x{\setminus}I$,
we add the one computed with the inheritance property to $\mathcal{P}_y$.
After the last round, we set $S_y$ to the union of the paths in $\mathcal{P}_y$.
This ensures $S_y \subseteq S_x$.
Intuitively, for an (inherited) $k$-path from $S_x$ to \emph{not} survive in $\mathcal{P}_y$, it must have been destroyed by \emph{all} the random deletions $I$.
If $x$ is a leaf, we not only store $S_x$ but also the path information in $\mathcal{P}_x$.

The query algorithm is  very similar to Algorithm~\ref{alg:tree-exact}.
The difference is that line~\ref{line:containment} now checks for $F \cap S_x \subseteq A_y$ (previously, $F \subseteq A_y$).
In line~\ref{line:result}, where a leaf $x$ is reached,
the query algorithm searches $\mathcal{P}_x$ for a $k$-path that is disjoint from $F$.
If one exists, it is output and the whole algorithm terminates;
if no such path exists, the search continues with the next tree.
If none of the trees $T_0, \dots, T_{K-1}$ produce such a path,
it is reported that $G{-}F$ does not have a $k$-path.

\paragraph{Analysis.}
We use different parameters for this oracle:\\
$h = \sqrt{f \ln (k/f)}$, $K = c \ 8^h f\nwspace \ln n$, $\alpha = (k/f)^{f/h}$, and $p = (f/k)^{1/h}$.
Note that $\alpha = 1/p^f$ still holds.

Let $y$ be a node in the tree $T_i$, $r$ its depth, and $S_x$ the edge set in the parent
($S_x = E$ in the root).
Assume $G{-}F$ has a $k$-path, let $P$ be the one computed with the inheritance property.
$E(P)$ is its set of links.
We say $y$ is \emph{well behaved} if it satisfies the following properties:
(P1)~$F \cap S_x\subseteq A_y$,
(P2)~$|E(P) \cap A_y|\leq p^r k$,
and (P3)~$P \in \mathcal{P}_y$.

The query algorithm enforces (P1) in every step of the way.
We are mainly interested in the probability of (P3) holding in the leave that is reached for a query.
We bound this using well-behaved children.

\begin{lemma}
	Let $i \in [K]$ and $x$ a non-leaf node in $T_i$.
	\begin{enumerate}[($i$)]
		\item If $x$ satisfies \emph{(P1)}, then there exists a child $y$ of $x$ also
			satisfying \emph{(P1)} with probability is at least $1-\frac{1}{e}$.
		\item If $x$ satisfies \emph{(P2)}, then, for any child $y$ of $x$,
		the probability that $y$ also satisfies \emph{(P2)} is at least $\frac{1}{4}$.
		\item If $x$ is well behaved and has a child $y$ that satisfies both \emph{(P1)} and \emph{(P2)},
			then the probability that $y$ is even well behaved (satisfies \emph{(P3)}) is at least $1-\frac{1}{h}$.
\end{enumerate}
	Moreover, the events are independent of each other.
\label{lem:well-behaved}
\end{lemma}

\begin{proof}
	Suppose $x$ satisfies (P1), let $S_z$ be the set of its parent (or $S_x = E$)
	and let $y_0,\ldots,y_{\alpha-1}$ be the children of $x$. 
	We have $F \cap S_x  \subseteq F \cap S_z \subseteq A_x$.
	The first inclusion is due to $S_x \subseteq S_z$ and the second one due to (P1).
	Since the elements of $A_{y_j}$ are sampled from $A_x$ with probability $p$,
	there exists an index $j \in [\alpha]$ such that $F \cap S_x \subseteq A_{y_j}$ with probability 
	$1-\prod_{j=0}^{\alpha-1} \prob[F \cap S_x \nsubseteq A_{y_j}] \geq 1-(1-p^f)^{\alpha} \geq 1-\frac{1}{e}$.

	We turn to Clause~($ii$).
	Let $r$ be the depth of $x$. 
	We now assume that $x$ satisfies (P2) (but not necessarily the other properties).
	That means, at most a $p^r$ ratio of the $k$ links of the path $P$ are removed in the network $G_x$ (are in the set $A_x$).
	For any child $y$ of $x$, the \emph{expected} size of $E(P) \cap A_y$ is $p \cdot |E(P) \cap A_x| \le p^{r+1} k$.
	In fact, the random variable $|E(P) \cap A_y|$ is binomially distributed with parameters $|E(P) \cap A_x|$ and $p$.
	The Central Limit Theorem states that $\prob[|E(P) \cap A_y| > p^{r+1} k] \le \frac{3}{4}$.
	Since the depth of $y$ is $r+1$, that implies that $y$ satisfies (P2) with probability $\frac{1}{4}$.
	
	The main part of this proof is Clause~($iii$) as it involves the new sets $S_y$.
	Let $x$ be well behaved and its child $y$ satisfy (P1) and (P2).
	Consider any of the round in the creation of $S_y$ and let $I \subseteq A_y$ be the sampled subset.
	We want to know whether the specific $k$-path $P$ from $G{-}F$ is included in $\mathcal{P}_y$ in this round.
	Due to the inheritance property, it is sufficient that $E(P) \subseteq (S_x{\setminus}I)$ 
	and at the same time $F \cap (S_x{\setminus}I) = \emptyset$.
	The latter condition formalizes that the network with edge set $S_x{\setminus}I$ is a subnetwork of $G{-}F$.

	We bound the probability of the two events.
	Parent $x$ is well behaved, a fortiori it satisfies (P3), 
	thus $E(P) \subseteq S_x$.
	Each link in $I$ is drawn from $A_y$ with probability $p^{h-(r+1)}/2$.
	We have $E(P) \subseteq (S_x{\setminus}I)$ if none of the links in $E(P) \cap A_y$ are drawn.
	Since $y$ satisfies (P2), this has probability at least	
	$(1-\frac{p^{h-(r+1)}}{2})^{|E(P) \cap A_y|} \ge (1-\frac{p^{h-(r+1)}}{2})^{p^{r+1}k}$.
	Inserting the definition $p = (\frac{f}{k})^{1/h}$, we get $p^{h-(r+1)} = \frac{f}{p^{r+1}k}$
	and $p^{r+1}k \ge p^h k = f$.
	Above estimate gives
	$(1-\frac{p^{h-(r+1)}}{2})^{p^{r+1}k} = (1-\frac{f}{2 \cdot p^{r+1}k} )^{p^{r+1}k} 
		\ge (1-\frac{f}{2 \cdot f} )^{f} = \frac{1}{2^f}$.

	Next is the probability that the sets $F$ and $S_x{\setminus}I$ are disjoint.
	Since the child $y$ also satisfies (P1), we have $F \cap S_x \subseteq A_y$.
	The sample $I$ is also a subset of $A_y$. 
	So for $F \cap S_x{\setminus}I = \emptyset$ all links in $F \cap S_x$ must be selected for $I$.
	This event has probability $(\frac{p^{h-(r+1)}}{2})^{|F \cap S_x|} \ge (\frac{p^{h-(r+1)}}{2})^f = \frac{p^{f(h-(r+1))}}{2^f}$.
	Recall that we chose the parameter such that $p^{f} = \alpha^{-1}$.
	The last estimate is thus equal to $2^{-f} \alpha^{-(h-(r+1))}$.
	
	Since the events $E(P) \subseteq (S_x{\setminus}I)$ and $F \cap (S_x{\setminus}I) = \emptyset$ are independent,
	they occur together with probability at least $4^{-f} \alpha^{-(h-(r+1))}$.
	We give the construction algorithm of our data structure
	$4^f \alpha^{h-(r+1)} \ln h$  rounds to try for it.
	The probability to include the path $P$ in $\mathcal{P}_y$ in any of the rounds is thus at least
	$1-\left(1-\frac{1}{4^f \alpha^{h-(r+1)}}\right)^{4^f\alpha^{h-(r+1)}\cdot \ln h} 
		\ge 1-\frac{1}{h}$.
\end{proof}

We now prove the correctness of the query algorithm.
If $G{-}F$ does not have a $k$-path, then the procedure indeed reports this fact.
If it reaches a leaf $x$ at all, it explicitly scans $\mathcal{P}_x$ for a path that is disjoint from $F$.
In the case that $G{-}F$ has a $k$-path, we argue over the parent-child traversals.
It is convenient to also include the algorithm jumping into the root of a tree as the first step.
Recall that for the root $y$, we set $A_y = E$ and use the convention $S_x = E$.
Therefore, $y$ trivially satisfies the properties (P1) and (P2).
Lemma~\ref{lem:well-behaved} thus shows that the root is well behaved with probability $(1-\frac{1}{e})\frac{1}{4}(1-\frac{1}{h})$.
Furthermore, since the algorithm actively looks for a child satifying (P1),
Lemma~\ref{lem:well-behaved} can also be iterated over the following $h$ traversals.
In summary a well-behaved child is reached with probability $((1-\frac{1}{e})\frac{1}{4}(1-\frac{1}{h}))^{h+1} \ge \frac{1}{8^h}$.
Repeating this in all $K = c \cdot 8^h f \ln n$ trees for a sufficiently large constant $c >0$
gives a high success probability over the $|\binom{E}{\le f}| = O(n^{2f})$ possible query sets $F$. 

\paragraph{Query Time, Space, and Preprocessing Time.}

Recall that we chose the parameters as
$K = c \nwspace 8^h f \ln n$, $\alpha = (k/f)^{f/h}$, and $h = \sqrt{f \ln(k/f)}$.
As before, this implies
$8^h = 8^{ \sqrt{f \ln(k/f)}} 
		= (k/f)^{\frac{\ln 8}{\ln(k/f)} \sqrt{f \ln(k/f)}} = (k/f)^{o(1)}$,
where we use $f = o(\log k)$.
That means, we have $K = (k/f)^{o(1)} f \ln n$.
However, the assumption on $f$ also implies that $\log f = o(\log k)$ (for any positive base of the logarithm).
Via $\log(k/f) = \log(k) - \log(f) = (1-o(1)) \log(k)$,
we get the seemingly stronger statement $f = o(\log(k/f))$.
We conclude $\alpha {=}(k/f)^{f/h} {=} (k/f)^{\sqrt{f/\ln(k/f)}} {=} (k/f)^{o(1)}$.

The total query time is $O(K \alpha hf + Kf4^f\log h)$.
The first term is derived as in the generic construction, assuming that the
test $F \cap S_x \subseteq A_y$ can be done in $O(f)$ time.
We explain below how to implement that.
The second term describes the time needed to search the collection $\mathcal{P}_x$ in all
the leaves the query algorithm may reach.
Using the parameters gives
$O((k/f)^{o(1)} f^2 \sqrt{f \ln(k/f)} \ln n 
		+ (k/f)^{o(1)} f^2 4^f \log h )$.
The second term of order $\Otilde(4^f f^{2-o(1)} k^{o(1)})$ is dominating.

For the space, consider a node $y$ in one of the trees that is not a root,
let $S_x$ be the second set associated with its parent.
We define $B_y = S_x \cap A_y$.
Observe that we have $F\cap S_ x \subseteq A_y$ if and only if $F\cap S_x\subseteq B_y$.
So it is enough to store $B_y$ instead of $A_y$
and we can indeed make that check in time $O(|F \cap S_x|) = O(f)$.
It will be advantageous for the analysis that both $B_y$ and $S_y$ are subsets of $S_x$.

The final data structure stores the trees $T_i$ for all $i \in [K]$;
in each root $x$ the set $S_x$; in each non-root node $y$ the sets $B_y$ and $S_y$;
and in each leaf $y$ the set $S_y$ and paths in $\mathcal{P}_y$.

Let $r$ be the depth of a node $x$, and let $y$ be a child of $x$.
It is enough to bound the size of the path collection $\mathcal{P}_x$
since it dominates the size of the union $S_x$,
and in turn the sizes of the subsets $B_y$ and $S_y$.
$\mathcal{P}_x$ contains at most $4^f \alpha^{h-r} \ln h$ different $k$-paths
and thus takes space $O(4^f \alpha^{h-r} k \log h)$.
Note that this is independent of the input graph $G$.

There are $\alpha^r$ many nodes in a tree at depth $r$,
giving space $\Otilde(4^f \alpha^{h} k)$ per level,
and $\Otilde(Kh 4^f \alpha^{h} k)$ for all levels and trees together.
Due to $\alpha^h = (k/f)^f$, this is
$\Otilde(4^f f^{3/2} (k/f)^{f+o(1)} k) = \Otilde(((k{+}4)/f)^{f+o(1)} f^{3/2} k)$.

For the preprocessing, we use the $k$-path algorithm by \citet{Williams_2009} running in time $2^k \nwspace \poly(n)$.
We spend time $4^f \alpha^{h-r} \cdot 2^k \nwspace \poly(n)$ preprocessing one tree node at depth $r$,
$4^f \alpha^{h} \cdot 2^k \nwspace \poly(n)$ in each level, and 
$((k{+}4)/f)^{f+o(1)} f^{3/2} \cdot 2^k \nwspace \poly(n)$ in total.

\paragraph{Cliques, Stars, and Cycles.}
Throughout the analysis, we never actually used the path structure of the solution,
only that it had $k$ edges. We can thus replace $k$-paths by any other graph motif we desire.
However, be aware that $k$-cliques have $\binom{k}{2} = O(k^2)$ edges.
Also, there are no known algorithms for finding $k$-cliques in time $g(k) \cdot \poly(n)$.

\section{Open Questions}
We made some progress in the design of fault-tolerant 
data structures for network problems. 
There are, however, several interesting open problems left in the area. 
\begin{itemize}
\item The randomized construction by \citet{WY13} has $\Otilde(fL^f)$ subnetworks.
	The deterministic one by \citet{KarthikParter21DeterministicRPC} needs $\Otilde(fL)^{f+1}$. 
	Is it possible to design a \emph{deterministic} $(L,f)$-replacement path covering
	with $\Otilde(L^f)$ subnetworks?
\item  \citet{KarthikParter21DeterministicRPC}
	also gave a lower bound on the size of $(L,f)$-RPCs.
	They showed that for any $n$, $L$, $f$ such that
	$(L/f)^{f+1} \le n$, there exists a network $G$ with $n$ nodes
	such that any $(L,f)$-RPC for $G$ must contain $\Omega((L/f)^f)$ networks.
	Even \citet{WY13} are off by a factor $\Otilde(f^{f+1})$.
	Can this gap be closed?
\item Beyond the mere number of subnetworks, our sampling trees requires $\Omega(n^2)$ space 
	when storing the networks associated with the nodes.
	Is there a more compact data structure for indexed subnetworks?
	It is known that in subquadratic space one must relax the requirement of 
	retrieving exact shortest paths ~\cite{ThorupZ05}.
\item Is it possible to generalize our framework to extremal distances, such as fault-tolerant diameter \cite*{Bilo22Extremal}, or to other network models, like temporal networks \cite*{Deligkas25EdgeCoverTemporalGraphs}? 
\item Finally, it would be interesting to implement our indexing scheme and compare its empirical performance with the \citet{WY13} construction. 
\end{itemize}

\newpage

\section{Acknowledgements}

In a previous version of this work, we claimed in Theorem~\ref{thm:k-path}
a preprocessing time of $(\frac{k+4}{f})^{f+o(1)} f^{3/2} \cdot 1.66^k \, \textup{\poly}(n)$ also for directed networks.
We are thankful to Dean Hirsch for pointing out this mistake.

\bibliography{DSO_aaai}

\begin{thebibliography}{41}
\providecommand{\natexlab}[1]{#1}

\bibitem[{Adamic(1999)}]{Adamic99thesmall}
Adamic, L.~A. 1999.
\newblock The Small World Web.
\newblock In \emph{Proceedings of the 3rd European Conference on Research and
  Advanced Technology for Digital Libraries (ECDL)}, 443--452.

\bibitem[{Afek et~al.(2002)Afek, Bremler{-}Barr, Kaplan, Cohen, and
  Merritt}]{Afek02RestorationbyPathConcatenation_journal}
Afek, Y.; Bremler{-}Barr, A.; Kaplan, H.; Cohen, E.; and Merritt, M. 2002.
\newblock {Restoration by Path Concatenation: Fast Recovery of MPLS Paths}.
\newblock \emph{Distributed Computing}, 15: 273--283.

\bibitem[{Albert, Jeong, and Barabási(1999)}]{Albert1999Diameter}
Albert, R.; Jeong, H.; and Barabási, A.-L. 1999.
\newblock Diameter of the World-Wide Web.
\newblock \emph{Nature}, 401: 130--131.

\bibitem[{Alman and Hirsch(2022)}]{AlmanHirsch22ExteriorAlgebras}
Alman, J.; and Hirsch, D. 2022.
\newblock {Parameterized Sensitivity Oracles and Dynamic Algorithms Using
  Exterior Algebras}.
\newblock In \emph{Proceedings of the 49th International Colloquium on
  Automata, Languages, and Programming (ICALP)}, 9:1--9:19.

\bibitem[{Alon, Galil, and Margalit(1997)}]{Alon97ExponentAPSP}
Alon, N.; Galil, Z.; and Margalit, O. 1997.
\newblock {On the Exponent of the All Pairs Shortest Path Problem}.
\newblock \emph{Journal of Computer and System Sciences}, 54: 255--262.

\bibitem[{Bil{\`{o}} et~al.(2023)Bil{\`{o}}, Chechik, Choudhary, Cohen,
  Friedrich, Krogmann, and Schirneck}]{BiloCCC0KS23}
Bil{\`{o}}, D.; Chechik, S.; Choudhary, K.; Cohen, S.; Friedrich, T.; Krogmann,
  S.; and Schirneck, M. 2023.
\newblock {Approximate Distance Sensitivity Oracles in Subquadratic Space}.
\newblock In \emph{Proceedings of the 55th Symposium on Theory of Computing
  (STOC)}, 1396--1409.

\bibitem[{Bil{\`{o}} et~al.(2024)Bil{\`{o}}, Chechik, Choudhary, Cohen,
  Friedrich, and Schirneck}]{Bilo24AImprovedDO}
Bil{\`{o}}, D.; Chechik, S.; Choudhary, K.; Cohen, S.; Friedrich, T.; and
  Schirneck, M. 2024.
\newblock {Improved Distance (Sensitivity) Oracles with Subquadratic Space}.
\newblock In \emph{Proceedings of the 65th Symposium on Foundations of Computer
  Science (FOCS)}, 1550--1558.

\bibitem[{Bilò et~al.(2022{\natexlab{a}})Bilò, Casel, Choudhary, Cohen,
  Friedrich, Lagodzinski, Schirneck, and
  Wietheger}]{Bilo22FixedParameterSensitivityOracles}
Bilò, D.; Casel, K.; Choudhary, K.; Cohen, S.; Friedrich, T.; Lagodzinski,
  J.~G.; Schirneck, M.; and Wietheger, S. 2022{\natexlab{a}}.
\newblock {Fixed-Parameter Sensitivity Oracles}.
\newblock In \emph{Proceedings of the 13th Innovations in Theoretical Computer
  Science Conference (ITCS)}, 23:1--23:18.

\bibitem[{Bilò et~al.(2022{\natexlab{b}})Bilò, Choudhary, Cohen, Friedrich,
  and Schirneck}]{Bilo22Extremal}
Bilò, D.; Choudhary, K.; Cohen, S.; Friedrich, T.; and Schirneck, M.
  2022{\natexlab{b}}.
\newblock {Deterministic Sensitivity Oracles for Diameter, Eccentricities and
  All Pairs Distances}.
\newblock In \emph{Proceedings of the 49th International Colloquium on
  Automata, Languages, and Programming (ICALP)}, 68:1--68:19.

\bibitem[{Björklund et~al.(2017)Björklund, Husfeldt, Kaski, and
  Koivisto}]{Bjoerklund17NarrowSieves}
Björklund, A.; Husfeldt, T.; Kaski, P.; and Koivisto, M. 2017.
\newblock {Narrow Sieves for Parameterized Paths and Packings}.
\newblock \emph{Journal of Computer and System Sciences}, 87: 119--139.

\bibitem[{Blum, Funke, and Storandt(2018)}]{Blum_Funke_Storandt_2018}
Blum, J.; Funke, S.; and Storandt, S. 2018.
\newblock Sublinear Search Spaces for Shortest Path Planning in Grid and Road
  Networks.
\newblock \emph{Proceedings of the AAAI Conference on Artificial Intelligence},
  32(1).

\bibitem[{Chechik et~al.(2017)Chechik, Cohen, Fiat, and Kaplan}]{ChCoFiKa17}
Chechik, S.; Cohen, S.; Fiat, A.; and Kaplan, H. 2017.
\newblock {($1{+}\varepsilon$)-Approximate \mbox{$f$-Sensitive} Distance
  Oracles}.
\newblock In \emph{Proceedings of the 28th Symposium on Discrete Algorithms
  (SODA)}, 1479--1496.

\bibitem[{Chechik et~al.(2012)Chechik, Langberg, Peleg, and Roditty}]{CLPR12}
Chechik, S.; Langberg, M.; Peleg, D.; and Roditty, L. 2012.
\newblock {$f$-Sensitivity Distance Oracles and Routing Schemes}.
\newblock \emph{Algorithmica}, 63: 861--882.

\bibitem[{Cho, Shin, and Oh(2024)}]{ChoS024}
Cho, K.; Shin, J.; and Oh, E. 2024.
\newblock {Approximate Distance Oracle for Fault-Tolerant Geometric Spanners}.
\newblock In \emph{Proceedings of the 38th AAAI Conference on Artificial
  Intelligence (AAAI)}, 20087--20095.

\bibitem[{Chung and Lu(2002)}]{chung2002average}
Chung, F.; and Lu, L. 2002.
\newblock {The Average Distances in Random Graphs with Given Expected Degrees}.
\newblock \emph{Proceedings of the National Academy of Sciences}, 99:
  15879--15882.

\bibitem[{Deligkas et~al.(2025)Deligkas, Döring, Eiben, Goldsmith, Skretas,
  and Tennigkeit}]{Deligkas25EdgeCoverTemporalGraphs}
Deligkas, A.; Döring, M.; Eiben, E.; Goldsmith, T.-L.; Skretas, G.; and
  Tennigkeit, G. 2025.
\newblock {How Many Lines to Paint the City: Exact Edge-Cover in Temporal
  Graphs}.
\newblock In \emph{Proceedings of the 39th AAAI Conference on Artificial
  Intelligence (AAAI)}.
\newblock To appear in the same proceedings.

\bibitem[{Demetrescu et~al.(2008)Demetrescu, Thorup, Chowdhury, and
  Ramachandran}]{DeThChRa08}
Demetrescu, C.; Thorup, M.; Chowdhury, R.~A.; and Ramachandran, V. 2008.
\newblock {Oracles for Distances Avoiding a Failed Node or Link}.
\newblock \emph{SIAM Journal on Computing}, 37: 1299--1318.

\bibitem[{Dey and Gupta(2024)}]{DeyGupta24}
Dey, D.; and Gupta, M. 2024.
\newblock Nearly Optimal Fault Tolerant Distance Oracle.
\newblock In Mohar, B.; Shinkar, I.; and O'Donnell, R., eds., \emph{Proceedings
  of the 56th Symposium on Theory of Computing (STOC)}, 944--955.

\bibitem[{Duan and Pettie(2009)}]{DuanP09a}
Duan, R.; and Pettie, S. 2009.
\newblock {Dual-Failure Distance and Connectivity Oracles}.
\newblock In \emph{Proceedings of the 20th Symposium on Discrete Algorithms
  (SODA)}, 506--515.

\bibitem[{Duan and Pettie(2010)}]{DuanP10}
Duan, R.; and Pettie, S. 2010.
\newblock {Connectivity Oracles for Failure Prone Graphs}.
\newblock In Schulman, L.~J., ed., \emph{Proceedings of the 42nd Symposium on
  Theory of Computing (STOC)}, 465--474.

\bibitem[{Duan and Pettie(2017)}]{DuanP:17}
Duan, R.; and Pettie, S. 2017.
\newblock {Connectivity Oracles for Graphs Subject to Vertex Failures}.
\newblock In \emph{Proceedings of the 28th Symposium on Discrete Algorithms
  (SODA)}, 490--509.

\bibitem[{Duan and Ren(2022)}]{DuRe22}
Duan, R.; and Ren, H. 2022.
\newblock {Maintaining Exact Distances under Multiple Edge Failures}.
\newblock In \emph{Proceedings of the 54th Symposium on Theory of Computing
  (STOC)}.
\newblock To appear.

\bibitem[{Duan, Wu, and Zhou(2023)}]{Duan23FMMviaAsymtericHashing}
Duan, R.; Wu, H.; and Zhou, R. 2023.
\newblock {Faster Matrix Multiplication via Asymmetric Hashing}.
\newblock In \emph{Proceedings of the 64th Symposium on Foundations of Computer
  Science (FOCS)}, 2129--2138.

\bibitem[{Friedrich and Krohmer(2018)}]{friedrich2018diameter}
Friedrich, T.; and Krohmer, A. 2018.
\newblock {On the Diameter of Hyperbolic Random Graphs}.
\newblock \emph{SIAM Journal on Discrete Mathematics}, 32: 1314--1334.

\bibitem[{Gu and Ren(2021)}]{GuR21}
Gu, Y.; and Ren, H. 2021.
\newblock Constructing a Distance Sensitivity Oracle in O(n{\^{}}2.5794 {M)}
  Time.
\newblock In Bansal, N.; Merelli, E.; and Worrell, J., eds., \emph{48th
  International Colloquium on Automata, Languages, and Programming, {ICALP}
  2021, July 12-16, 2021, Glasgow, Scotland (Virtual Conference)}, volume 198
  of \emph{LIPIcs}, 76:1--76:20. Schloss Dagstuhl - Leibniz-Zentrum f{\"{u}}r
  Informatik.

\bibitem[{Hofstad(2016)}]{Hofstad16RandomGraphs}
Hofstad, R. v.~d. 2016.
\newblock \emph{{Random Graphs and Complex Networks}}, volume~1 of
  \emph{Cambridge Series in Statistical and Probabilistic Mathematics}.
\newblock Cambridge, UK: Cambridge University Press.

\bibitem[{Karthik and Parter(2021)}]{KarthikParter21DeterministicRPC}
Karthik, C.~S.; and Parter, M. 2021.
\newblock {Deterministic Replacement Path Covering}.
\newblock In \emph{Proceedings of the 32nd Symposium on Discrete Algorithms
  (SODA)}, 704--723.

\bibitem[{Kleinberg(2000)}]{kleinberg2000navigation}
Kleinberg, J.~M. 2000.
\newblock {Navigation in a Small World}.
\newblock \emph{Nature}, 406: 845--845.

\bibitem[{Nešetřil and Poljak(1985)}]{NesetrilPoljak85SubgraphProblem}
Nešetřil, J.; and Poljak, S. 1985.
\newblock {On the Complexity of the Subgraph Problem}.
\newblock \emph{Commentationes Mathematicae Universitatis Carolinae}, 26:
  415–419.

\bibitem[{Ouyang et~al.(2020)Ouyang, Yuan, Qin, Chang, Zhang, and
  Lin}]{OuyangYQCZL20}
Ouyang, D.; Yuan, L.; Qin, L.; Chang, L.; Zhang, Y.; and Lin, X. 2020.
\newblock Efficient Shortest Path Index Maintenance on Dynamic Road Networks
  with Theoretical Guarantees.
\newblock \emph{Proceedings of the VLDB Endowment}, 13(5): 602--615.

\bibitem[{Patrascu and Thorup(2007)}]{PatrascuT:07}
Patrascu, M.; and Thorup, M. 2007.
\newblock Planning for Fast Connectivity Updates.
\newblock In \emph{Proceedings of the 48th Symposium on Foundations of Computer
  Science (FOCS)}, 263--271.

\bibitem[{Seidel(1995)}]{Seide95APSPUnweightedUndirected}
Seidel, R. 1995.
\newblock {On the All-Pairs-Shortest-Path Problem in Unweighted Undirected
  Graphs}.
\newblock \emph{Journal of Computer and System Sciences}, 51: 400--403.

\bibitem[{Shoshan and Zwick(1999)}]{ShZw99}
Shoshan, A.; and Zwick, U. 1999.
\newblock {All Pairs Shortest Paths in Undirected Graphs with Integer Weights}.
\newblock In \emph{Proceedings of the 40th Symposium on Foundations of Computer
  Science (FOCS)}, 605--615.

\bibitem[{Thorup and Zwick(2005)}]{ThorupZ05}
Thorup, M.; and Zwick, U. 2005.
\newblock {Approximate Distance Oracles}.
\newblock \emph{Journal of the ACM}, 52: 1--24.

\bibitem[{Tsur(2019)}]{Tsur_2019}
Tsur, D. 2019.
\newblock {Faster Deterministic Parameterized Algorithm for $k$-Path}.
\newblock \emph{Theoretical Computer Science}, 790: 96–104.

\bibitem[{Vassilevska~Williams et~al.(2024)Vassilevska~Williams, Xu, Xu, and
  Zhou}]{VWilliams24MatrixMultiplicationAlphaToOmega}
Vassilevska~Williams, V.; Xu, Y.; Xu, Z.; and Zhou, R. 2024.
\newblock {New Bounds for Matrix Multiplication: from Alpha to Omega}.
\newblock In \emph{Proceedings of the 35th Symposium on Discrete Algorithms
  (SODA)}, 3792--3835.

\bibitem[{Watts and Strogatz(1998)}]{watts_collective_1998}
Watts, D.~J.; and Strogatz, S.~H. 1998.
\newblock Collective dynamics of ‘small-world’ networks.
\newblock \emph{Nature}, 393: 440--442.

\bibitem[{Weimann and Yuster(2013)}]{WY13}
Weimann, O.; and Yuster, R. 2013.
\newblock {Replacement Paths and Distance Sensitivity Oracles via Fast Matrix
  Multiplication}.
\newblock \emph{ACM Transactions on Algorithms}, 9: 14:1--14:13.

\bibitem[{Williams(2009)}]{Williams_2009}
Williams, R.~R. 2009.
\newblock {Finding Paths of Length \(k\) in \(O^*(2^k)\) Time}.
\newblock \emph{Information Processing Letters}, 109: 315–318.

\bibitem[{Zhang, Li, and Zhou(2021)}]{Zhang0021}
Zhang, M.; Li, L.; and Zhou, X. 2021.
\newblock An Experimental Evaluation and Guideline for Path Finding in Weighted
  Dynamic Network.
\newblock \emph{Proceedings of the VLDB Endowment}, 14(11): 2127--2140.

\bibitem[{Zwick(2002)}]{Zwick02DirectedAPSP}
Zwick, U. 2002.
\newblock {All Pairs Shortest Paths Using Bridging Sets and Rectangular Matrix
  Multiplication}.
\newblock \emph{Journal of the ACM}, 49: 289–317.

\end{thebibliography}

\end{document}